\documentclass[runningheads]{llncs}
\usepackage{graphicx}
\usepackage{amsmath}
\usepackage{amssymb}
\usepackage{dsfont}
\usepackage[mathscr]{euscript}

\newcommand{\Lf}{\mathcal{L}}
\newcommand{\pset}{\mathcal{P}}
\newcommand{\bx}{\bar{x}}
\newcommand{\dl}{\delta}
\newcommand{\dla}{\delta^{\rm a}}
\newcommand{\dls}{\delta^{\rm s}}

\newcommand{\ra}{r^{\rm a}}
\newcommand{\rs}{r^{\rm s}}
\newcommand{\psets}{\pset^{\rm s}}
\newcommand{\pseta}{\pset^{\rm a}}

\newcommand{\xa}{x^{\rm a}}
\newcommand{\bell}{\bar{\ell}}

\DeclareMathOperator*{\argmax}{arg\,max}

\newcommand{\perv}{{\rm PI}}

\newcommand{\vs}{\vspace{-1mm}}

\newcommand{\gee}{\mathcal{G}}

\newcommand{\poa}{\ensuremath{{\rm PoA}}}

\newcommand{\Lfopt}{\ensuremath{\mathcal{L}^*}}

\newcommand{\Lfnf}{\mathcal{L}^{\rm nf}}

\begin{document}
\title{Partial Altruism is Worse than Complete Selfishness in Nonatomic Congestion Games\thanks{This work was supported by NSF Grant \#ECCS-2013779.}}
\titlerunning{Partial Altruism is Worse than Complete Selfishness}
%
\author{Philip N. Brown\inst{1}\orcidID{0000-0003-3953-0503} }
\authorrunning{P. N. Brown}
%
\institute{University of Colorado Colorado Springs, USA 80918\\
\email{pbrown2@uccs.edu}}
\maketitle              
\begin{abstract} 
We seek to understand the fundamental mathematics governing infrastructure-scale interactions between humans and machines, particularly when the machines’ intended purpose is to influence and optimize the behavior of the humans.
To that end, this paper investigates the worst-case congestion that can arise in nonatomic network congestion games when a fraction of the traffic is completely altruistic (e.g., benevolent self-driving cars) and the remainder is completely selfish (e.g., human commuters).
We study the worst-case harm of altruism in such scenarios in terms of the \emph{perversity index}, or the worst-case equilibrium congestion cost resulting from the presence of altruistic traffic, relative to the congestion cost which would result if all traffic were selfish.
We derive a tight bound on the perversity index for the class of series-parallel network congestion games with convex latency functions, and show three facts: First, the harm of altruism is maximized when exactly half of the traffic is altruistic, but it gracefully vanishes when the fraction of altruistic traffic approaches either 0 or 1.
Second, we show that the harm of altruism is linearly increasing in a natural measure of the ``steepness'' of network latency functions.
Finally, we show that for any nontrivial fraction of altruistic traffic, the harm of altruism exceeds the price of anarchy associated with all-selfish traffic: in a sense, partial altruism is worse than complete selfishness.

\end{abstract}

\section{Introduction}

As technology becomes ever-more integrated with society, it becomes increasingly urgent to understand and optimize the relationship between human social behavior and the technical performance of infrastructure systems~\cite{Ratliff2018}.
An important test case for this is the problem of routing in transportation networks; it is well-known that if individuals select routes across a city to minimize their own transit delay, this can easily lead to suboptimal traffic congestion~\cite{Roughgarden2005}.
A popular model of such systems is the nonatomic network congestion game, which models traffic as a continuum of infinitely-many infinitesimally-small self-interested agents.
Much research on nonatomic congestion games has investigated how to mitigate the inefficiencies resulting from self-interested behavior, with approaches including monetary incentives~\cite{Ferguson2020a}, centralized routing schemes~\cite{Roughgarden2004}, and information provision~\cite{Cheng2016}.
Another active line of research has sought to understand how perceptual biases in users' costs might affect the quality of the resulting aggregate behavior~\cite{Chen2014,Kleer2017,Meir2015b,Nikolova2015,Sekar2019,Gupta2019} (we provide more detail in Section~\ref{ssec:background}).

In this paper, inspired by the potential afforded by emerging paradigms such as the Internet of Things, we ask the following question: suppose a planner were able to design the behaviors of some group of system users (perhaps a group of self-driving cars).
How should those users be designed to behave? 
Specifically, what cost functions should be applied to guide those users' routing choices?
While many potential cost functions are possible, this paper initiates a study on the impact of endowing a subset of users with \emph{altruistic} cost functions, or ones which lead each user to consider the effects of her actions on those around her.
The key motivation for our study on altruism is this: if a planner could design \emph{all} user cost functions, then altruistic cost functions would in some sense be the unique best choice as they induce optimal Nash equilibria in nonatomic congestion games irrespective of network topology.

Accordingly, the model of heterogeneous altruism proposed in~\cite{Chen2014} is of special relevance to our paper.
In this model, users are diverse in their altruism, each potentially weighing her impact on others in an individual way.
In our context, the central result of~\cite{Chen2014} is that on simple networks (specifically parallel networks on which all users have access to all paths), altruism behaves ``as intended'' -- that is, increasing the fraction of users that are altruistic always improves worst-case equilibrium congestion.
Unfortunately, recent work has highlighted the fact that if transportation networks are sufficiently complex, heterogeneous altruism need not lead to improvements in aggregate behavior; put differently, altruism can be \emph{perverse} on some networks~\cite{Brown2017b,Brown2020b,Sekar2019}.

Thus, to understand when and if artificial altruism can be used to improve routing performance, it is imperative that its potential negative effects be fully characterized.
In this paper, for the class of series-parallel networks with convex cost functions, we investigate the quality of equilibria resulting when a fraction $\rs\in[0,1]$ of traffic is fully selfish and the remainder is fully altruistic.
We assess the potential harm of partial altruism by a metric known as the \emph{perversity index}, defined as the ratio between the equilibrium congestion costs of a heterogeneous population and a fully-selfish population, taken in worst-case over a class of games.
That is, if the perversity index is $1$, partial altruism never induces equilibria that are worse than their fully-selfish counterparts; on the other hand, if the perversity index is large, partial altruism can substantially degrade equilibria relative to fully-selfish equilibria.
We prove the following in Theorem~\ref{thm:main} and its Corollary~\ref{cor1}:
\begin{enumerate}
\item For all classes of latency functions, the worst-case perversity index (i.e., harm of altruism) occurs when $\rs=1/2$; that is, when there is an even split between altruistic and selfish traffic.
\item We provide an expression for the perversity index that is increasing, linear, and unbounded in the ``steepness'' of a network's latency functions (for polynomial latency functions, the steepness is essentially the maximum degree; note however that our result holds for any convex latency function).
\item Our bounds are tight for all $\rs$, and all of our worst-case instances are 3-link parallel networks.
\end{enumerate}
Figure~\ref{fig:plots} plots our bounds for polynomial latency functions of degree~$p\in\{1,\dots,4\}$; note that for each of these, the maximum perversity index is $1+p/2$, and degrades gracefully to $1$ as $\rs\to0$ and $\rs\to1$.
Finally, note that for each $p$, the maximum harm of partial altruism is actually \emph{greater} than the price of anarchy.%
\footnote{
The price of anarchy is the worst-case ratio between the fully-selfish equilibria and the optimal allocation~\cite{Roughgarden2005}.
}
In particular note that for linear-affine latency functions ($p=1$), the perversity index can be as high as $1.5$, exceeding the well-known price of anarchy bound of $4/3$ for the same class of games.
Our Corollary~\ref{cor2} shows that this is ubiquitous: for any $\rs\in(0,1)$, there exists a choice of cost functions (possibly very steep) such that partial altruism causes more harm than selfishness in worst case.

\begin{figure}
	\centering
	\includegraphics[width=0.95\textwidth]{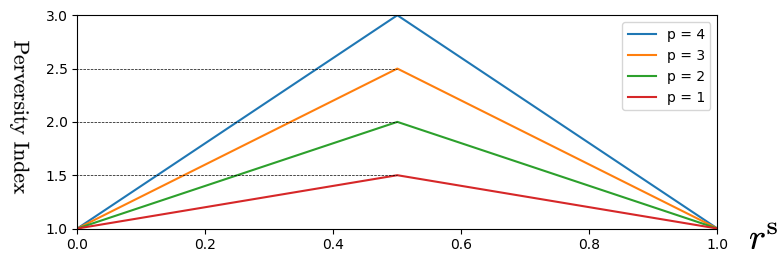} 
	\caption{Perversity index with respect to selfish fraction $\rs$, plotted for polynomial latency functions of degree $p\in\{1,\dots,4\}$ (see Theorem~\ref{thm:main} and Corollary~\ref{cor1}).
	}
	\label{fig:plots}\vs\vs\vs\vs\vs
\end{figure}

\section{Model and Related Work}

\subsection{Routing Problem}

Consider a network routing problem for a network $(V,E)$ comprised of vertex set $V$ and edge set $E$.
There is a unit mass of traffic that needs to be routed from a common source $s$ to a common destination $t$.
In this work we consider that traffic is composed of two types, which we term \emph{selfish} and \emph{altruistic}, with masses $\rs$ and $\ra$, respectively, with $\ra+\rs=1$.
We write $\pset\subseteq 2^E$ to denote the set of \emph{paths} available to traffic, where each path $p\in\pset$ consists of a set of edges connecting $s$ to $t$. 
The selfish (resp., altruistic) traffic can access path set $\psets$ (resp., $\pseta$), where $\psets$ and $\pseta$ are arbitrary subsets of $\pset$.
We say that a graph is \emph{series-parallel} if it is any of 
\begin{enumerate}
\item a single edge,
\item two series-parallel graphs connected in series, or
\item two series-parallel graphs connected in parallel.
\end{enumerate}

For each type $ y\in\{{\rm a},{\rm s}\}$, we write $x^y_p\geq0$ to denote the mass of traffic of type $y$ using path $p\in\pset^y$.
A \emph{feasible flow for type $y$} is written  $x^y\in\mathbb{R}_{\geq0}^{|\pset^y|}$, and is an assignment of $r^y$ units of traffic to paths in $\pset^y$ such that $\sum_{p\in\pset}x_p^y = r^y$.
A \emph{feasible flow} $x\in\mathbb{R}_{\geq0}^{|\pset|}$ is an assignment of traffic to paths such that $x_p:=\sum_{y:p\in\pset^y}x_p^y$, where the corresponding flows $x^y$ are each feasible for their corresponding types.

Given a flow $x$, the flow on edge $e$ is given by $x_e = \sum_{p:e\in p}x_p$.
As before, we may denote the flow of type $y\in\{{\rm a},{\rm s}\}$ on edge $e$ by $x_e^y$.
To characterize transit delay as a function of traffic flow, each edge $e\in E$ is associated with a specific latency function $\ell_e:[0,1]\rightarrow[0,\infty)$, where $\ell_e(x_e)$ denotes the delay experienced by users of edge $e$ when the edge flow is $x_e$.
We adopt the standard assumptions that each latency function is nondecreasing, convex, and continuously differentiable.
We measure the cost of a flow $x$ by the \emph{total latency}, given by
\begin{equation}
\Lf(x) =  \sum_{e\in E}  x_e  \ell_e(x_e)=  \sum\limits_{p\in \pset}  x_p  \ell_p(x), \label{eq:totlatfpath}
\end{equation}
where $\ell_p(x) := \sum_{e\in p}\ell_e(x_e)$ denotes the latency on path $p$ given flow $x$.

Given edge latency function $\ell_e(x_e)$, the \emph{marginal cost} function associated with $\ell_e$ is denoted $\ell_e^{\rm mc}$ and is given by
\begin{equation}
\ell_e^{\rm mc}(x_e) = \ell_e(x_e) + x_e\ell'_e(x_e),
\end{equation}
where $\ell'$ denotes the flow derivative of $\ell$.
%
A \emph{routing problem} is fully specified by $G=\left(V,E,\left\{\ell_e\right\},\pseta,\psets,\rs\right)$.
We write $\gee(\gamma)$ to denote the set of all nonatomic congestion games on series-parallel networks having latency functions with marginal-cost ratio bounded by~$\gamma$; that is, for every $G\in\gee(\gamma)$ and every $\ell_e\in G$, $\ell_e^{\rm mc}(x_e)\leq\gamma\ell_e(x_e)$ for all $x_e\in[0,1]$.

\subsection{Heterogeneous Routing Game}

To understand the potential negative effects of uncoordinated heterogeneous altruistic behavior, we model the above routing problem as a non-atomic congestion game.
We consider each type of traffic to be composed of infinitely-many infinitesimal users, and the cost function of each user is determined by its type.
Given a flow $x$, the cost that a selfish-type user experiences for using path ${p} \in \psets$ is simply the latency of the path:
\begin{equation} \label{eq:selfishdef}
J^{\rm s}_p(x) := \sum\limits_{e\in {p}}\ell_{e}(x_{e}).
\end{equation}

The cost that an altruistic-type user experiences for using path ${p} \in \pseta$ is the marginal cost of the path:
\begin{equation} \label{eq:mcdef}
J^{\rm a}_p(x) = \ell^{\rm mc}_p(x) := \sum\limits_{e\in {p}} \left[\ell_{e}(x_{e}) + x_e\ell_e'(x_e)\right].
\end{equation}
Prior literature justifies this formulation as an appropriate cost function for altruistic users; we refer readers to~\cite{Chen2014} for a detailed exposition.
Informally, note that an altruistic user's link cost is the sum of two components: first, the latency of the link; second, the user's marginal effect on other users of the link -- an altruistic user fully accounts for the negative congestion externalities she imposes on others.
In addition, these altruistic cost functions are fundamental in the literature on taxation for behavior influence in congestion games~\cite{Brown2020b}.

We assume that each user selects the lowest-cost path from the available source-destination paths.
We call a flow $f$ a \emph{Nash flow} if all users are individually using minimum-cost paths given the choices of other users.
That is, the following condition is satisfied for each type $y\in\{{\rm a},{\rm s}\}$:
\begin{equation}
\forall p,p'\in\pset^{y},\ x^y_p>0 \implies J_p^{y}(x)\leq J_{p'}^{y}(x). \label{eq:nfdef}
\end{equation}
It is well-known that a Nash flow exists for any non-atomic congestion game of the above form \cite{Mas-Colell1984}.
Note that a Nash flow must satisfy the well-known Wardrop condition that for each type $y\in\{{\rm a},{\rm s}\}$, all used paths have equal cost~\cite{Wardrop1952}:
\begin{equation}
\forall p,p'\in\pset^{y},\ x^y_p>0 \mbox{ and } x^y_{p'}>0 \implies J_p^{y}(x)= J_{p'}^{y}(x). \label{eq:wardrop}
\end{equation}

This paper frequently considers the effect of ``converting'' all altruistic traffic to selfish traffic.
To that end, given routing problem $G$, we write $\bar{G}$ to denote the \emph{homogenized} version of $G$ in which all traffic behaves selfishly.
The homogenized $\bar{G}$ is identical to $G$ except that the altruistic traffic adopts the selfish cost function~\eqref{eq:selfishdef}; all path sets, latency functions, traffic rates, and network topology remain the same.
When $x$ is a Nash flow for $G$, we often write $\bx$ to denote a Nash flow for $\bar{G}$.

\subsection{Performance Metrics: Price of Anarchy and Perversity Index}

For ease of notation, we write $\Lfnf(G,\rs)$ to denote the total latency of a worst-case heterogeneous Nash flow for routing problem $G$ and selfish fraction $\rs$ (with corresponding altruistic fraction $\ra=1-\rs$).
Let $\Lfnf\left(\bar{G}\right)$ denote the total latency of a Nash flow on homogenized routing problem $\bar{G}$.
Let $\Lfopt(G)$ denote the total latency of an optimal flow on $G$ (note that $\Lfopt(G)$ has no dependency on $\rs$, since it considers only the time cost of flows).
In our context, we define the \emph{price of anarchy} of a class of games $\gee$ with selfish fraction $\rs$ as the worst-case ratio of the cost of a heterogeneous Nash flow with that of an optimal flow, formally written
\begin{equation} \label{eq:poadef}
\poa\left(\gee,\rs\right) \triangleq \sup_{G\in\gee} \frac{\Lfnf(G,\rs)}{\Lfopt(G)}.
\end{equation}
Note that when $\rs=1$ (i.e., all traffic is selfish),~\eqref{eq:poadef} coincides with the standard definition of price of anarchy~\cite{Roughgarden2005}.
At the other extreme, when $\rs=0$ (i.e., all traffic is altruistic), it is well-known that $\poa\left(\gee,0\right)=1$; that is, homogeneous altruistic Nash flows are always optimal~\cite{Chen2014}.

In this paper, we study a related metric known as the \emph{perversity index} which aptly captures the worst-case effects of heterogeneous altruism; this metric was originally proposed to study the risks associated with marginal-cost taxation~\cite{Brown2020b}.
The perversity index captures the harm that altruism causes \emph{relative to selfishness}; it is thus similar in many respects to the \emph{deviation ratio} of~\cite{Kleer2017} and the price of risk aversion of~\cite{Piliouras2013}.
This is a natural modification of the price of anarchy concept obtained by replacing the optimal total latency with the all-selfish total latency $\Lfnf(\bar{G})$:
\begin{equation} \label{eq:pidef}
\perv\left(\gee,\rs\right) \triangleq \sup_{G\in\gee} \frac{\Lfnf(G,\rs)}{\Lfnf\left(\bar{G}\right)}.
\end{equation}

Here, if $\gee$ has a large perversity index, this means that networks exist in $\gee$ for which heterogeneous (partially-altruistic) Nash flows can be considerably worse than corresponding homogeneous all-selfish Nash flows.
Note that for any $\gee$, it holds that $\perv(\gee,0)=1$, since $\rs=0$ leads to optimal flows (and some optimal flows are also homogeneous selfish Nash flows); it also trivially holds that $\perv(\gee,1)=1$.
Furthermore, the perversity index is upper-bounded by the price of anarchy: $\perv(\gee,\rs)\leq\poa(\gee,\rs)$; this is because on any $G$, $\Lfnf(\bar{G})\geq\Lfopt(G)$.

\subsection{Related Work} \label{ssec:background}


\subsubsection{Altruism and Biased Costs in Congestion Games.}
The effects of altruism in nonatomic congestion games are investigated in detail in~\cite{Chen2014}, where it is shown that in parallel networks (i.e., 2-node networks) in which all users can access all paths (in our model, $\pseta=\psets=\pset$), the price of anarchy due to heterogeneous altruism is always less than the price of anarchy for homogeneous selfish populations.
This was extended by~\cite{Brown2017b} which showed that on the same class of networks, the perversity index due to heterogeneous altruism is unity; that is, not only is altruism helpful in worst-case, it is also helpful on every network.
It was also shown in~\cite{Brown2017b} that more complex non-parallel networks exist which do not possess this property, even when $\pseta=\psets=\pset$; on these pathological networks altruism can cause harm relative to corresponding all-selfish flows.
However, the question of how much harm has remained an open question.

Our work is tightly connected with the broader theme of ``biased'' congestion games, where players misinterpret edge cost functions in some systematic way.
Several works have investigated the price of anarchy under various payoff biases such as pessimism~\cite{Meir2015b}, risk-aversion~\cite{Nikolova2015}, and uncertainty~\cite{Sekar2019}.
Analogous to our definition of the perversity index, the authors of~\cite{Lianeas2016} study the ``price of risk aversion,'' which measures how society's risk preferences affect aggregate congestion as compared to ordinary Nash flows.
Similarly,~\cite{Kleer2017} studies the ``deviation ratio,'' which measures essentially the same quantity for arbitrary cost function biases.
However, as we discuss in our technical results, the bounds given by~\cite{Kleer2017} do not apply to a significant portion of the parameter space required to analyze the perversity index of altruistic populations.

\subsubsection{Marginal-Cost Pricing in Congestion Games.}
A great deal of work has focused on improving the price of anarchy through the use of \emph{marginal-cost pricing,} a road tolling scheme which charges users of each edge $e$ a price of
\begin{equation}
\tau_e^{\rm mc}\left(x_e\right)  = x_e  \ell_e'(x_e), \label{eq:mc}
\end{equation}
with the aim of internalizing each user's congestion externality.
Marginal-cost pricing intends to induce an artificial altruism in users' costs, and it has long been known that if all users have a known unit price-sensitivity (i.e., all trade off time and money equally), marginal-cost pricing leads to optimal Nash flows irrespective of network topology~\cite{Beckmann1956,Sandholm2002}.
However, these optimality guarantees vanish in the presence of user heterogeneity; in the extreme case when some users ignore tolls, it is known that marginal-cost pricing is perverse and can induce Nash flows that are considerably worse than the corresponding uninfluenced selfish flows~\cite{Brown2020b}.
Thus, the present paper's results regarding the harm caused by heterogeneous altruistic users immediately imply corresponding results regarding the potential harm caused by imperfect marginal-cost pricing.

\subsubsection{Stackelberg Routing.}
One of the goals of the present paper is to understand how socially-networked autonomous agents (such as self-driving cars) should be designed to optimally interact with and influence human agents.
In light of this, it is important to mention the related notion of Stackelberg routing.
In a Stackelberg routing scenario, a system planner has complete routing control over some fraction of network traffic, and the question is how to route this traffic in such a way that when the remaining traffic routes selfishly, the resulting congestion cost is minimized.
A variety of results exist in this framework, including positive results for simple networks~\cite{Roughgarden2004} and negative results for complex networks~\cite{Bonifaci2010}.

\section{Contribution: Altruists Cause Bounded Harm}

One might hope that increasing the fraction of altruists in a routing game would result in improvements in the quality of resulting Nash flows in a similar way that it does in parallel networks~\cite{Brown2020b}; unfortunately, in this paper we show that this need not be the case.
Nonetheless, for classes of routing games with a bounded marginal-cost ratio, we also show that the harm resulting from heterogeneous altruism is bounded and achieves its worst-case when $\rs=1/2$.
Our main result, capturing these two facts, is given as the following theorem.

\begin{theorem} \label{thm:main}
Let $\gee(\gamma)$ be the class of series-parallel nonatomic congestion games with cost functions having maximum marginal-cost ratio $\gamma\geq1$, and let $\rs\in[0,1]$ be the fraction of selfish traffic.
Then 
\begin{equation} \label{eq:main pi}
\perv\left(\gee(\gamma),\rs\right) = \left\{
\begin{array}{ll}
1+\rs(\gamma-1) 	&\ \mbox{if }\ \rs\leq 1/2 \\
\gamma-\rs(\gamma-1) 	&\ \mbox{if }\ \rs> 1/2.
\end{array}\right.
\end{equation}
In particular, it always holds that 
\begin{equation}
\perv(\gee(\gamma,\rs)) \leq \perv(\gee(\gamma,1/2)) = \frac{1+\gamma}{2}. \label{eq:main bound}
\end{equation}
\end{theorem}

The proof of Theorem~\ref{thm:main} is presented in detail in Section~\ref{sec:proof}; before its presentation, we note here several salient details and discuss the consequences of the result.
First, the functional form of the perversity index in~\eqref{eq:main pi} is strikingly simple, being linear in $\gamma$ and piecewise-linear in $\rs$.
For any $\gamma$, its maximum is attained at $\rs=1/2$: that is, worst-case outcomes occur for equal parts selfish and altruistic traffic irrespective of the functional form of cost functions.
As we show in Lemma~\ref{lem:tight}, the bound is tight for any fixed $\gamma$ and $\rs$, in the sense that given $\gamma\geq1$ and $\rs\in[0,1]$, there exists a choice of latency function $\ell$ satisfying $\ell^{\rm mc}\leq \gamma\ell$ which can be used to construct networks which achieve the bound.

A intuitive and commonly-studied class of latency functions is the \emph{polynomial} latency functions of degree $p$; the following corollary grounds Theorem~\ref{thm:main} by deriving the perversity index associated with polynomial latency functions.
\begin{corollary} \label{cor1} 
Let $\gee^p$ be the class of two-terminal series-parallel nonatomic congestion games with polynomial cost functions of maximum degree $p\geq1$.
Then for every $p\in\mathbb{N}$ and every $\rs\in[0,1]$, the following holds:
\begin{equation} \label{eq:poly pi}
\perv\left(\gee(p),\rs\right) = \left\{
\begin{array}{ll}
1+p\rs 	&\ \mbox{if }\ \rs\leq 1/2 \\
1+p(1-\rs) 	&\ \mbox{if }\ \rs> 1/2.
\end{array}\right.
\end{equation}
\end{corollary}

\begin{proof}
It suffices to show that if $\ell(x)$ is a polynomial of degree $p$, then $\ell^{\rm mc}(x)\leq(p+1)\ell(x)$; then~\eqref{eq:poly pi} follows directly from~\eqref{eq:main pi}.
To that end, let $\ell(x)=\sum_{i=0}^p a_ix^i$, where $a_i\geq0$ for all $i$. 
Then applying~\eqref{eq:mcdef}, we have
\begin{align}
\ell^{\rm mc}(x) 	&= \sum_{i=0}^p \left( a_ix^i + i a_ix^i\right) \nonumber \\
				&\leq \sum_{i=0}^p (p+1)a_ix^i  \nonumber \\
				&= (p+1)\ell(x).  \nonumber
\end{align}
\hfill\qed
\end{proof}

Note that the expression in~\eqref{eq:poly pi} is $\Theta(p)$ for all $\rs\notin\{0,1\}$.
This contrasts sharply with canonical price of anarchy results for the same class of networks which state that the all-selfish homogeneous $\poa(\gee(p),1)=\Theta(p/\log p)$.
In essence, Corollary~\ref{cor1} indicates that in a worst-case sense, partial altruism is \emph{worse} than uniform selfishness.

We sharpen this observation in our next corollary, which demonstrates that this phenomenon is ubiquitous:
\begin{corollary} \label{cor2} 
Let $\gee(p)$ be the class of two-terminal series-parallel nonatomic congestion games with polynomial cost functions of maximum degree $p\geq1$.
For any $p$, let $r^*$ be defined as
\begin{equation} \label{eq:fbounds}
r^*= \frac{1}{(p+1)^{\frac{p+1}{p}}-p}=\Theta\left(\frac{1}{\log p}\right). 
\end{equation}
Then for every $p$ and every $\rs\in(r^*,1-r^*)$, the following holds:
\begin{equation}
\poa(\gee(p),\rs) > \poa(\gee(p),1). \label{eq:poaineq}
\end{equation}
\end{corollary}

\begin{proof}
Here, we leverage the following well-known expression (see~\cite{Roughgarden2003} for details):
\begin{equation}
\poa(\gee(p),1) = \frac{1}{1-p(p+1)^{-(p+1)/p}}=\Theta\left(\frac{p}{\log p}\right). \label{eq:poa}
\end{equation}
Recall that $\poa(\gee(p),\rs)\geq\perv(\gee(p),\rs)$ for all $\rs$, and note that $\perv(\gee(p),\rs)=\perv(\gee(p),1-\rs)$.
Let $r^*$ be the solution to $1+pr^*=\poa(\gee(p),1)$; in can be shown that this solution always occurs on the interval $[0,1/2]$.
Thus, due to the functional form of~\eqref{eq:poly pi},~\eqref{eq:poaineq} holds for any $\rs\in(r^*,1-r^*)$.
\hfill\qed
\end{proof}

This result is striking, since it implies that for very steep cost functions, even a very small fraction of selfish traffic is sufficient to drive the heterogeneous price of anarchy above the homogeneous price of anarchy, and our Lemma~\ref{lem:tight} shows that this can be realized on very simple 3-link networks.
Note that Corollary~\ref{cor2} is similar in spirit to a Stackelberg routing result of~\cite[Theorem 3.1]{Bonifaci2010} which can be used to lower-bound the price of anarchy of heterogeneous Nash flows in our context.
In~\cite{Bonifaci2010}, the authors show that the price of anarchy is unbounded for any fixed fraction $\rs>0$ of selfish traffic.
However, that paper left open an important question: is the price of anarchy with heterogeneous traffic at least \emph{improved} relative to the all-selfish price of anarchy?
Corollary~\ref{cor2} answers this open question negatively: if we place no restriction on allowable cost functions, then for every $\rs\in(0,1)$, the price of anarchy is (1) unbounded, and (2) greater than if all traffic were selfish.

\section{Proof of Theorem~\ref{thm:main}} \label{sec:proof}

\subsection{Overview}
The proof of Theorem~\ref{thm:main} is completed in three steps, each with an associated lemma.
\begin{enumerate}
\item First, Lemma~\ref{lem:low rs ub} derives an upper bound on the perversity index for the case that $\rs\leq1/2$.
\item Second, Lemma~\ref{lem:hi rs ub} derives an upper bound on the perversity index for the case that $\rs>1/2$.
We note that Lemma~\ref{lem:hi rs ub} is also a consequence of~\cite[Theorem~1]{Kleer2017}, but we state it in terms of our specific model and provide a proof that is highly simplified relative to that of~\cite{Kleer2017}.
\item Finally, Lemma~\ref{lem:tight} provides a family of problem instances which exhibit a matching lower bound for Lemmas~\ref{lem:low rs ub} and~\ref{lem:hi rs ub}, completing the proof.
\end{enumerate}
We state the lemmas here and then proceed with their proofs in subsequent subsections.

\begin{lemma} \label{lem:low rs ub}
Let $G$ be a series-parallel nonatomic congestion game with selfish fraction $\rs\in[0,1]$ and maximum marginal-cost ratio $\gamma$.
If $x$ is a heterogeneous Nash flow for $G$ and $\bx$ is an all-selfish Nash flow for $\bar{G}$, the following holds:
\begin{equation}
\Lf(x)\leq\left((1-\rs)+\rs\gamma\right)\Lf(\bx). \label{eq:low rs ub}
\end{equation}
\end{lemma}

Lemma~\ref{lem:low rs ub} exhibits the intuitive feature that if a very small fraction of traffic is selfish (i.e., $\rs$ is very close to $0$), this small amount of selfish traffic can cause very little harm.
In the limit as $\rs\to0$, this resolves to the well-known fact that if $\rs=0$, then $\Lf(x)\leq\Lf(\bx)$ (since a homogeneous altruistic Nash flow is optimal).
However, the bound in~\eqref{eq:low rs ub} is demonstrably loose for large $\rs$ (e.g., $\ra=0$), since it resolves to the trivially loose bound of $\Lf(x)\leq\gamma\Lf(\bx)$ where both $x$ and $\bx$ are all-selfish Nash flows.%
\footnote{In any nonatomic congestion game satisfying our assumptions, every homogeneous selfish Nash flow is well-known to have the same total cost~\cite{Cole2003,Sandholm2002}.} %
Fortunately, the following Lemma~\ref{lem:hi rs ub} addresses the case of large $\rs$ and offers a degree of symmetry with respect to the bound in Lemma~\ref{lem:low rs ub}.

\begin{lemma} \label{lem:hi rs ub}
Let $G$ be a series-parallel nonatomic congestion game with selfish fraction $\rs\in[0,1]$ and maximum marginal-cost ratio $\gamma$.
If $x$ is a heterogeneous Nash flow for $G$ and $\bx$ is an all-selfish Nash flow for $\bar{G}$, the following holds:
\begin{equation}
\Lf(x)\leq\left(\rs+(1-\rs)\gamma\right)\Lf(\bx). \label{eq:hi rs ub}
\end{equation}
\end{lemma}

Much like before, Lemma~\ref{lem:hi rs ub} illustrates the intuitive feature that if a very \emph{large} fraction of traffic is selfish (i.e., $\rs$ is very close to $1$), the overall Nash flow behaves very much like a homogeneous selfish Nash flow.
Our final lemma shows that these bounds are tight by constructing a family of networks which achieve the bounds.

\begin{lemma} \label{lem:tight}
Let $\gamma>1$ and $\rs\in[0,1]$.
Then there exists a network $G\in\gee(\gamma)$ such that 
\begin{equation}
\frac{\Lfnf(G,\rs)}{\Lfnf\left(\bar{G}\right)}\geq \left\{
\begin{array}{ll}
1+\rs(\gamma-1) 	&\ \mbox{if }\ \rs\leq 1/2 \\
\gamma-\rs(\gamma-1) 	&\ \mbox{if }\ \rs> 1/2.
\end{array}\right.
\end{equation}
\end{lemma}

\subsection{Proofs of Supporting Lemmas}

First, note that the proofs of Lemmas~\ref{lem:low rs ub} and~\ref{lem:hi rs ub} require the following technical result.
We first state the following lemma, borrowed from~\cite[Proposition 2]{Sekar2019} and framed specifically in our context (this can also be deduced from~\cite[Lemma 2]{Milchtaich2006}):
\begin{lemma}[Sekar et al., 2019] \label{lem:borrowed}
Let $G$ be a series-parallel nonatomic congestion game.
If $x$ is a heterogeneous Nash flow for $G$ and $\bx$ is an all-selfish Nash flow for $\bar{G}$, the total cost experienced by selfish traffic in $x$ is no higher than that in $\bx$:
\begin{equation}
\sum_{p\in\pset} x^{\rm s}_p\ell_p(x) \leq \sum_{p\in\pset} \bx^{\rm s}_p\ell_p(\bx).
\end{equation}
\end{lemma}

\noindent With this result in hand, we are ready to provide the proof of Lemma~\ref{lem:low rs ub}.
In the following proofs, when $x$ is a Nash flow, we write $\Lambda(x)$ to denote the common latency experienced by members of the selfish type in $x$.\\

\noindent\emph{Proof of Lemma~\ref{lem:low rs ub}.}
Let $G$, $x$, and $\bx$ be as specified in the statement of Lemma~\ref{lem:low rs ub}.
The proof proceeds via a linearization argument enabled by the convexity of $\Lf$.
For each edge $e\in E$, let $\dl_e:=x_e-\bx_e$, so that we have
\begin{align}
\Lf(x) -\Lf(\bx) 	&= \sum_{e\in E} \left( x_e\ell_e(x_e)-\bx_e\ell_e(\bx_e) \right) \nonumber \\
			&= \sum_{e\in E} \left( (\bx_e+\dl_e)\ell_e(\bx_e+\dl_e)-\bx_e\ell_e(\bx_e)  \right) \nonumber \\
			&= \sum_{e\in E} \dl_e\left( \ell_e(\bx_e+\dl_e) + \frac{\bx_e}{\dl_e}(\ell_e(\bx_e+\dl_e)-\ell_e(\bx_e) ) \right). \nonumber \\
			&\leq \sum_{e\in E} \dl_e\left( \ell_e(\bx_e+\dl_e) + {\bx_e}\ell'_e(\bx_e+\dl_e)  \right), \label{eq:break1}
\end{align}
Where~\eqref{eq:break1} holds because each $\ell_e$ is convex and nondecreasing.
Note that since each $\bx_e+\dl_e\geq0$, (and thus $\dl_e(\bx_e+\dl_e)\geq\bx_e$), we may continue the estimate as:
\begin{align}
\Lf(x) -\Lf(\bx) 	&\leq \sum_{e\in E} \dl_e\left( \ell_e(\bx_e+\dl_e) + (\bx_e+\dl_e)\ell'_e(\bx_e+\dl_e)  \right) 	\nonumber \\
			&= \sum_{e\in E} \dl_e \ell_e^{\rm mc}(\bx_e+\dl_e)   	\nonumber \\
			&= \sum_{e\in E} \dl_e \ell_e^{\rm mc}(x_e). \label{eq:break2}
\end{align}
Define $\dl_{e,p}=\dl_p$ if $e\in p$, and $\dl_{e,p}=0$ otherwise, so that~\eqref{eq:break2} continues as
\begin{align}
\Lf(x) -\Lf(\bx) 	&\leq \sum_{e\in E} \sum_{p\in\pset} \dl_{e,p} \ell_e^{\rm mc}(x_e) \nonumber \\
			&= \sum_{p\in\pset} \dl_{p}  \ell_p^{\rm mc}(x) \nonumber \\
			&= \sum_{p\in\pset} \dls_p\ell_p^{\rm mc}(x) + \sum_{p\in\pset} \dla_p\ell_p^{\rm mc}(x),  \label{eq:break3}
\end{align}
Where for path $p$, let $\dla_p:=x_p^{\rm a}-\bx_p^{\rm a}$ and $\dls_p:=x_p^{\rm s}-\bx_p^{\rm s}$ so that $\dl_p = \dls_p+\dla_p$.
Let $\bar{p}:=\argmax_{p\in\pset^{\rm s}}\ell_p^{\rm mc}(x)$, let $\pset^+:=\{p\in\pset: \dls_p\geq0\}$, and let $\pset^-:=\{p\in\pset: \dls_p<0\}$.
Then since $\ell_p^{\rm mc}(x)\geq\ell_p(x)$ and $x^{\rm s}_{\bar{p}}>0$, we have
\begin{align}
\sum_{p\in\pset} \dls_p\ell_p^{\rm mc}(x)	&= \sum_{p\in\pset^+}\dls_p\ell_p^{\rm mc}(x) + \sum_{p\in\pset^-}\dls_p\ell_p^{\rm mc}(x) \nonumber \\
									&\leq \sum_{p\in\pset^+}\dls_p\ell_{\bar{p}}^{\rm mc}(x) + \sum_{p\in\pset^-}\dls_p\ell_{\bar{p}}(x) \nonumber \\
									&\leq \ell_{\bar{p}}(x)\bigg(\gamma\sum_{p\in\pset^+}\dls_p + \sum_{p\in\pset^-}\dls_p \bigg) \nonumber \\
									&= \Lambda^{\rm s}(x)\bigg(\gamma\sum_{p\in\pset^+}\dls_p - \sum_{p\in\pset^+}\dls_p \bigg) \nonumber \\
									&\leq \Lambda^{\rm s}(x)\rs(\gamma-1). \label{eq:break3.5}
\end{align}

Furthermore, note that by the definition of altruistic traffic at heterogeneous Nash flow, if $\dla_p>0$ and $\dla_q>0$ (respectively, $\dla_q<0$) for any paths $p,q$, it holds that $\ell_p^{\rm mc}(x)=\ell_q^{\rm mc}(x)$ (respectively, $\ell_p^{\rm mc}(x)\leq\ell_q^{\rm mc}(x)$) so that $\sum_{p\in\pset} \dla_p\ell_p^{\rm mc}(x)\leq0$.
Thus, continuing from~\eqref{eq:break3} and applying~\eqref{eq:break3.5}, we have
\begin{align}
\Lf(x) 	&\leq \Lf(\bx) + \rs(\gamma-1)\Lambda^{\rm s}(x) \nonumber \\
		&\leq \Lf(\bx) + \rs(\gamma-1)\Lambda^{\rm s}(\bx), \label{eq:break4}
\end{align}
where~\eqref{eq:break4} follows from Lemma~\ref{lem:borrowed}.
Writing $\Lf(\bx)=(\ra+\rs)\Lambda^{\rm s}(\bx)$, we thus have that
\begin{align}
\Lf(x)		&\leq (\ra+\rs)\Lambda^{\rm s}(\bx) + \rs(\gamma-1)\Lambda^{\rm s}(\bx) \nonumber \\
		&= \frac{\ra+\rs\gamma}{\ra+\rs}\Lf(\bx),
\end{align}
completing the proof of Lemma~\ref{lem:low rs ub}.
\hfill\qed

\noindent The proof of Lemma~\ref{lem:hi rs ub} proceeds similarly; here, we show that the total cost experienced by altruistic traffic in any Nash flow is upper bounded by a simple function of $\gamma$; we then use this fact to complete the proof.\\

\noindent\emph{Proof of Lemma~\ref{lem:hi rs ub}.}
We first derive the following simple upper bound on the cost experienced by altruistic traffic in $x$.
In the following, let path $\bar{p}$ be such that $\bx_{\bar{p}}^{\rm a}>0$, i.e., any path used by altruistic traffic in $\bx$.
\begin{align}
\sum_p x_p^{\rm a}\ell_p(x)	&\leq \sum_p x_p^{\rm a}\ell^{\rm mc}_p(x)	 		\nonumber \\
						&= \sum_p x_p^{\rm a}\ell^{\rm mc}_p(x) 			\nonumber \\
						&\leq  \ell^{\rm mc}_{\bar{p}}(x)\sum_p x_p^{\rm a} 	\nonumber \\ 
						&\leq  \ra\gamma\ell_{\bar{p}}(x)					\nonumber \\ 
						&\leq \ra\gamma\Lambda^{\rm s}(\bx). 			\label{eq:break5}
\end{align}
With~\eqref{eq:break5} in hand, we are ready to compute the bound:
\begin{align}
\Lf(x) 	&= \sum_p x_p^{\rm s}\ell_p(x) + \sum_p x_p^{\rm a}\ell_p(x) 		\nonumber \\
		&\leq \rs\Lambda^{\rm s}(\bx) + \ra\gamma\Lambda^{\rm s}(\bx) 	\label{eq:bound1} \\
		&= \frac{\rs+\ra\gamma}{\rs+\ra}\Lf(\bx). 						\label{eq:bound2}
\end{align}
Inequality~\eqref{eq:bound1} follows from Lemma~\ref{lem:borrowed} and~\eqref{eq:break5}, and~\eqref{eq:bound2} follows from the fact that $\bx$ is an all-selfish Nash flow, and thus all traffic experiences the common latency $\Lambda^{\rm s}(\bx)$.
\hfill \qed

\noindent We end with a proof of Lemma~\ref{lem:tight}, giving a family of example networks whose Nash flow total latencies achieve the bound given in Lemmas~\ref{lem:low rs ub} and~\ref{lem:hi rs ub} for all admissible $\gamma$ and $\rs$.
Interestingly, a single network topology consisting of only 3 links suffices to prove the matching lower bound, indicating that the perversity index (much like the price of anarchy) has a degree of independence from network topology~\cite{Roughgarden2003}.\\

\noindent\emph{Proof of Lemma~\ref{lem:tight}.}
Let $\rs\in[0,1]$ (so that $\ra=1-\rs$) and let $\gamma\geq 1$.
Consider the following nonatomic congestion game, also depicted generally in Figure~\ref{fig:tight}.
This game occurs on a parallel network of 3 edges, so we shall refer to edges as \emph{paths}, write the path set $\pset=\{1,2,3\}$, and write the latency of path $i$ as $\ell_i(x_i)$.
Let altruists have access only to paths $1$ and $2$, and selfish traffic have access to all paths, or $\pseta=\{1,2\}$ and $\psets=\pset$.
Let $r:=\max\{\rs,\ra\}$, so that $1-r\leq \min\{\rs,\ra\}$.
Let latency function $\bell$ be chosen%
\footnote{
Note that this is always possible subject to mild conditions on the considered class of latency functions; specifically, we require that $\ell(x)$ implies the existence of $\alpha\ell(x)$ and $\ell(\alpha x)$ for all $\alpha>0$.
}
such that $\bell(r)=1$ and $\bell^{\rm mc}(r)=\gamma$.
The path latency functions are selected to be $\ell_1(x_1)=\gamma$, $\ell_2(x_2)=\bell(x_2)$, and $\ell_3(x_3)=1$ as depicted in Figure~\ref{fig:tight}.

\begin{figure}
	\centering
	\includegraphics[width=0.95\textwidth]{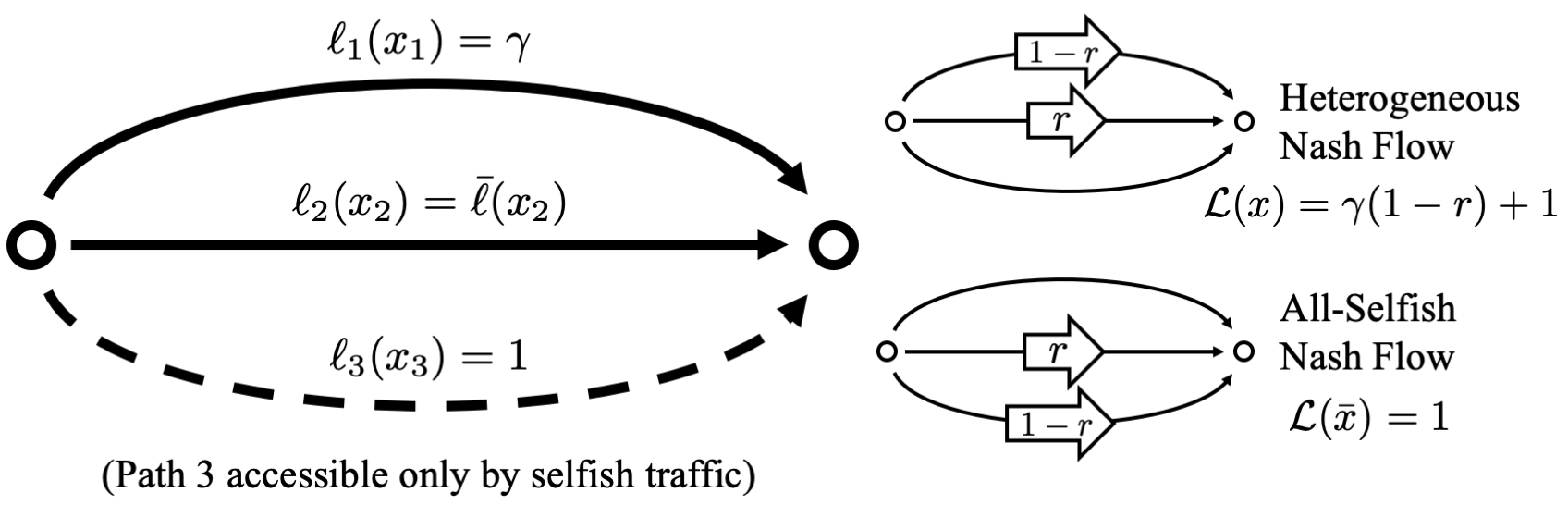}
	\caption{Network used to generate matching lower bound in Lemma~\ref{lem:tight}.}
	\label{fig:tight}\vs\vs\vs\vs
\end{figure}
First, we claim that the flow $x:=(1-r,r,0)$ with $\xa_1=1-r$ is a Nash flow on this instance.
To see this, first note that this flow is feasible, since $1-r\leq \ra$ by the definition of $r$ (see above), and by the fact that selfish traffic has access to all paths.
Next, by our choice of $\bell$, the altruistic traffic incentive condition is satisfied:
\begin{equation*}
\ell^{\rm mc}_1(1-r) = \gamma \leq \bell^{\rm mc}(r) = \ell^{\rm mc}_2(r).
\end{equation*}
Again by the choice of $\bell$, the selfish traffic incentive condition between paths $2$ and $3$ is satisfied:
\begin{equation*}
\ell_2(r) = \bell(r) \leq 1 = \ell_3(0).
\end{equation*}
Finally, note that with this choice of latency functions, for all $x_2'\leq r$, selfish traffic never uses path $1$ in any Nash flow, since $\ell_1(x_1)=\gamma>1\geq\ell_2(x_2')$.
Straightforward computation shows that 
\begin{equation} \label{eq:lower bound}
\Lf(x)=(1-r)\gamma + r = \left\{
\begin{array}{ll}
\rs\gamma+(1-\rs) 	&\ \  \mbox{if } \rs\leq 1/2 \\
(1-\rs)\gamma+\rs 	&\ \  \mbox{if } \rs> 1/2.
\end{array}
\right.
\end{equation}

Having computed the total latency of a heterogeneous Nash flow, we turn now to the computation of a corresponding homogeneous selfish Nash flow.
We claim that $\bx=(0,r,1-r)$ is a homogeneous selfish Nash flow.
As before, $\bx$ is feasible because $1-r\leq\rs$.
Since $\bell(r)=1$, all of the selfish traffic is indifferent between paths $2$ and $3$, and strictly prefers them to path $1$.
All traffic experiences a latency of $1$, so $\Lf(\bx)=1$, so the desired lower bound is given directly by~\eqref{eq:lower bound}.
\hfill \qed\vs\vs\vs

\section{Conclusions}

We have shown the first tight bounds on the perversity index for heterogeneous altruistic populations for series-parallel networks.
Future work will focus on extending these bounds to populations with diverse altruism levels, and exploring how these arguments can be extended and generalized to arbitrary cost function biases.

\bibliographystyle{splncs04}
\bibliography{../../library/library}
\end{document}